\documentclass[a4paper,onecolumn,nopdfoutputerror]{quantumarticle}

\usepackage{amsmath,amssymb,amsfonts,amsthm,bm}
\usepackage{graphicx}
\usepackage{quantikz}
\usepackage[numbers,sort&compress]{natbib}
\usepackage{hyperref}

\graphicspath{{./}}

\newtheorem{proposition}{Proposition}[section]
\newtheorem{lemma}[proposition]{Lemma}
\newtheorem{theorem}[proposition]{Theorem}
\newtheorem{corollary}[proposition]{Corollary}

\newcommand{\ii}{\mathrm{i}}
\newcommand{\su}{\mathfrak{su}}
\newcommand{\Lie}{\mathrm{Lie}}

\begin{document}

\title{Restricting Trainable Lie-Algebra Growth in Equivariant Quantum Networks via Hierarchical Ancilla-Controlled Subspace Projections}
\author{Ting Li}
\email{lit@njupt.edu.cn}
\affiliation{Nanjing University of Posts and Telecommunications, Nanjing, China}
\author{Zhiming Xiao}
\affiliation{Nanjing University of Posts and Telecommunications, Nanjing, China}
\author{Qibiao Tang}
\affiliation{Nanjing University of Posts and Telecommunications, Nanjing, China}

\begin{abstract}
Equivariant quantum networks encode symmetry as an inductive bias, which can improve generalization and may also favor optimization convergence. Equivariance alone, however, does not constrain the noncommuting closure of trainable generators, and this closure can still grow rapidly in symmetry-preserving variational circuits. We introduce a hierarchical ancilla-controlled architecture that addresses this Lie-algebra-growth mechanism. Commuting invariant-sector projectors on the data register select parameterized operations on a shared ancilla register, where the noncommuting trainable dynamics is confined. The trainable circuit decomposes into compatible joint sectors, giving a sector-probability-weighted ancilla response and an explicit view of parameter sharing across hierarchical paths. For an ancilla dimension $d_A=2^m$ and $K_\ell$ retained layer-wise control modes, we prove the group-independent bound $\dim(\mathfrak g)\le (d_A^2-1)\prod_{\ell=1}^{L}(K_\ell+1)$. The bound is polynomial in the number of data qubits when $m$ and $K_\ell$ remain constant along a logarithmic-depth hierarchy. Particle-number and parity projectors illustrate the general construction, while a fixed Clebsch--Gordan coupling tree supplies a concrete $SU(2)$ realization with rotation-invariant scalar outputs. Finite-size state-vector simulations exhibit slower gradient-variance decay and larger initialization gradients than generic and conventional rotationally equivariant circuits over the studied system sizes. The same realization fits sparse rotation-invariant classification tasks and geometry-dependent Heisenberg ground-state energies. These results support restricted trainable Lie-algebra growth as a structural strategy for initialization trainability in the regimes considered here.
\end{abstract}


\maketitle

\section{Introduction}

Variational quantum circuits support a broad range of near-term algorithms in quantum simulation and quantum learning \cite{peruzzo_variational_2014,kandala_hardware_2017,preskill_quantum_2018,benedetti_parameterized_2019,cerezo_variational_2021,schuld_introduction_2015,biamonte_quantum_2017,havlicek_supervised_2019,schuld_circuit_2020}. Their performance depends strongly on circuit structure. Symmetry-aware constructions reduce the hypothesis space to functions compatible with geometric or physical priors, often improving generalization and parameter efficiency. Several studies also report favorable optimization behavior when equivariance removes redundant directions \cite{sauvage_group-invariant_2022,meyer_exploiting_2023,schatzki_theoretical_2024,west_provably_2024,gil-fuster_opportunities_2024}. These benefits make equivariant quantum networks a natural starting point for structured variational models.

Small gradients and loss concentration remain common as circuit size and expressivity increase \cite{mcclean2018barren,cerezo2021cost,anschuetz_quantum_2022,larocca_diagnosing_2022,grant_initialization_2019,thanasilp_subtleties_2023,larocca_review_2024}. Equivariance constrains how a circuit transforms under a group, yet this constraint does not determine the size of the dynamical Lie algebra generated by its trainable operators. A symmetry-preserving ansatz may therefore explore a large operator manifold and retain a barren-plateau-related concentration mechanism. Lie-algebraic analyses connect this algebra to controllability, overparameterization, quantum Fisher information rank, and gradient concentration \cite{larocca_theory_2023,fontana2023adjoint,goh_lie-algebraic_2025,holmes_connecting_2022,wiersema_classification_2024,ragone_lie_2024,diaz_showcasing_2023,gil-fuster_relation_2025}. The architectural problem considered here is how to preserve the useful symmetry bias while restricting the noncommuting trainable closure.

We introduce a hierarchical ancilla-controlled equivariant quantum network for this purpose. A family of commuting invariant-sector projectors acts on the data register and supplies coherent control conditions. Parameterized noncommuting evolution acts on a shared ancilla register of fixed dimension. Temporary label qubits extract the relevant sector predicates and are uncomputed after each controlled operation. Products of commuting projectors describe the data-side Lie closure, while the non-Abelian component remains localized to the ancilla subsystem. The same joint projectors decompose the full circuit into compatible paths: input-dependent sector probabilities are paired with ancilla responses whose parameters are shared across paths.

The algebraic construction applies to a general group whenever the control projectors admit a common projective refinement and commute with the group representation. We derive a group-independent containment theorem for the resulting dynamical Lie algebra and a dimension upper bound. Constant ancilla size and a constant number of retained modes per layer give polynomial growth along a logarithmic-depth hierarchy. We then construct an $SU(2)$ realization from a fixed Clebsch--Gordan coupling tree. The circuit evolution is rotation equivariant; an invariant initial state, a trivial ancilla representation, and ancilla-only measurements produce rotation-invariant scalar outputs. Rotationally equivariant architectures based on Fourier, Schur--Weyl, and spin-network structure provide the relevant context for this realization \cite{west_provably_2024,zheng_speeding_2023,east_all_2023,zheng_sncqa_2023}, while horizontal-gate constructions illustrate a complementary route to symmetry-informed expressivity \cite{wiersema_geometric_2024}.

The numerical study uses finite-size state-vector simulations under matched resource scales. It probes whether the algebraic restriction appears in initialization-gradient statistics and whether the restricted $SU(2)$ circuit retains sufficient capacity for representative invariant learning tasks. Dynamical Lie algebra dimension serves here as a structural diagnostic of one trainability-limiting mechanism; gradient behavior also depends on the loss, initialization, data distribution, and observable.

The contributions are summarized as follows.
\begin{enumerate}
  \item We formulate a hierarchical ancilla-controlled architecture in which commuting invariant-sector projectors select noncommuting trainable operations on a shared ancilla register, and we separate its variational parameter count from symbolic extraction and control costs.
  \item We derive a global joint-sector decomposition that expresses the output through sector probabilities and parameter-shared path responses, clarifying how retained sectors and ancilla size regulate capacity.
  \item We prove a group-independent containment theorem for the trainable dynamical Lie algebra and derive dimension bounds for general and truncated-control regimes.
  \item We give $U(1)$ and $\mathbb Z_2$ instances and develop an $SU(2)$ realization from a fixed Clebsch--Gordan tree, establishing rotation-equivariant circuit evolution with rotation-invariant scalar readout.
\end{enumerate}

\section{Preliminaries}

\subsection{Symmetry-Compatible Quantum Models}

Let a group $G$ act on inputs as $x\mapsto g\cdot x$ and on the data Hilbert space through a unitary representation $\rho_D(g)$. A symmetry-compatible encoding satisfies
\begin{equation}\label{eq:encoding-equiv-en}
U_E(g\cdot x)=\rho_D(g)U_E(x)\rho_D(g)^\dagger.
\end{equation}
A trainable unitary $U_{\bm\theta}$ is equivariant when
\begin{equation}\label{eq:trainable-equiv-en}
[U_{\bm\theta},\rho_D(g)]=0,
\qquad g\in G.
\end{equation}
The circuit dynamics in this work satisfies Eq.~\eqref{eq:trainable-equiv-en}. Scalar invariance additionally uses a $G$-invariant initial state and a measurement that commutes with the representation. For an ancilla register carrying the trivial representation, any measurement confined to that register has this property. Data encoding strongly affects both expressivity and trainability in such models \cite{perez-salinas_data_2020,schuld_circuit_2020,jerbi_quantum_2023,du_learnability_2021}.

\subsection{Trainability and Dynamical Lie Algebra}

For Hermitian trainable generators $\{H_j\}$, the dynamical Lie algebra is
\begin{equation}\label{eq:dla-definition-en}
\mathfrak g=\Lie\{-\ii H_j\}.
\end{equation}
It contains the initial infinitesimal directions and the directions obtained from their nested commutators. The Baker--Campbell--Hausdorff relation
\begin{equation}
e^{\epsilon A}e^{\epsilon B}e^{-\epsilon A}e^{-\epsilon B}
=
e^{\epsilon^2[A,B]+O(\epsilon^3)},
\end{equation}
shows how a commutator direction appears from short evolutions. The dimension and representation structure of $\mathfrak g$ constrain the reachable unitary orbit and the maximal quantum Fisher information rank. Algebras approaching $\su(2^N)$ describe highly expressive evolution and are associated with concentration for broad classes of randomly initialized circuits and global costs \cite{larocca_diagnosing_2022,larocca_theory_2023,fontana2023adjoint,barthe_gradients_2024,sim_expressibility_2019,abbas_power_2021,holmes_connecting_2022,mhiri_constrained_2024}. We use the growth of $\mathfrak g$ to diagnose this structural mechanism. The observed gradient scale remains dependent on the circuit distribution, loss, data, and measurement.

\subsection{Relation to Existing Equivariant Architectures}

Group-invariant and equivariant variational models commonly enforce symmetry through commuting trainable layers, invariant observables, group averaging, or symmetry-adapted bases \cite{sauvage_group-invariant_2022,meyer_exploiting_2023,gil-fuster_opportunities_2024}. For rotations, Fourier- and Schur-based constructions organize states and gates by irreducible-representation data, while spin-network circuits use recoupling structure directly \cite{west_provably_2024,zheng_speeding_2023,east_all_2023,zheng_sncqa_2023}. Selected Fourier-based families also admit explicit trainability guarantees \cite{west_provably_2024}. Horizontal quantum gates pursue a related geometric objective through directions on homogeneous spaces and relax exact layer-wise equivariance in settings where that constraint limits accessible motion \cite{wiersema_geometric_2024}.

Within this landscape, the present construction combines invariant projective controls with a shared ancilla that carries the noncommuting trainable response. The resulting joint-sector decomposition makes the trainable Lie closure amenable to a group-independent containment bound. The realized expressivity and gradient behavior still depend on the chosen generators, encoding, depth, and loss.

\section{General Hierarchical Ancilla-Controlled Architecture}

\subsection{Registers and Projective Controls}

Let $\mathcal H_D$ be the data Hilbert space and
\begin{equation}
\mathcal H_A\cong\mathbb C^{d_A},
\qquad d_A=2^m,
\end{equation}
a shared register of $m$ ancilla qubits. A temporary register $\mathcal H_T$ supports coherent label extraction and returns to its reference state after each controlled operation. The data register carries the representation $\rho_D(g)$; the ancilla registers carry the trivial representation.

Layer $\ell$ retains $K_\ell$ joint control modes,
\begin{equation}
\Pi_{\ell,1},\ldots,\Pi_{\ell,K_\ell},
\end{equation}
from a projective decomposition of $\mathcal H_D$. The remaining modes form
\begin{equation}\label{eq:complement-en}
\Pi_{\ell,\perp}
=
I_D-\sum_{s=1}^{K_\ell}\Pi_{\ell,s}.
\end{equation}
They obey
\begin{equation}\label{eq:same-layer-orthogonal-en}
\Pi_{\ell,a}\Pi_{\ell,b}
=
\delta_{ab}\Pi_{\ell,a},
\qquad
\sum_{a\in\{1,\ldots,K_\ell,\perp\}}\Pi_{\ell,a}=I_D.
\end{equation}
The count $K_\ell$ refers to retained label patterns for the entire layer. It does not count sectors assigned independently to each block.

A common projective refinement supplies cross-layer compatibility:
\begin{equation}\label{eq:cross-commute-en}
[\Pi_{\ell,a},\Pi_{r,b}]=0
\qquad
\text{for all }\ell,r,a,b.
\end{equation}
The symmetry condition is
\begin{equation}\label{eq:projector-group-invariant-en}
[\Pi_{\ell,a},\rho_D(g)]=0,
\qquad g\in G.
\end{equation}
The complementary sector remains in the state space and passes through the layer without activating a parameterized ancilla operation.

\subsection{Controlled Ancilla Dynamics}

Choose traceless Hermitian operators $\{\Lambda_q\}_{q=1}^{d_A^2-1}$ such that $\{-\ii\Lambda_q\}$ is a basis of $\su(d_A)$. The Hamiltonian associated with retained mode $s$ is
\begin{equation}\label{eq:Hs-en}
H_{\ell,s}(\bm\theta_{\ell,s})
=
\sum_{q=1}^{d_A^2-1}
\theta_{\ell,s,q}\Lambda_q .
\end{equation}
The controlled layer is
\begin{equation}\label{eq:Ul-en}
U_\ell(\bm\theta_\ell)
=
\exp\left[
-\ii\sum_{s=1}^{K_\ell}
\Pi_{\ell,s}\otimes H_{\ell,s}(\bm\theta_{\ell,s})
\right],
\end{equation}
Set $G_\ell=\sum_s\Pi_{\ell,s}\otimes H_{\ell,s}$. Projector orthogonality gives
\begin{equation}
G_\ell^r
=
\sum_{s=1}^{K_\ell}\Pi_{\ell,s}\otimes H_{\ell,s}^r,
\qquad r\ge1.
\end{equation}
The power-series expansion of $e^{-\ii G_\ell}$ then yields
\begin{equation}\label{eq:Ul-block-en}
U_\ell(\bm\theta_\ell)
=
\sum_{s=1}^{K_\ell}
\Pi_{\ell,s}\otimes e^{-\ii H_{\ell,s}(\bm\theta_{\ell,s})}
+
\Pi_{\ell,\perp}\otimes I_A .
\end{equation}

The full circuit is
\begin{equation}\label{eq:full-circuit-en}
U(x;\bm\theta)
=
U_L(\bm\theta_L)\cdots U_1(\bm\theta_1)U_E(x).
\end{equation}
All layers reuse $\mathcal H_A$. Its final state depends on the ordered sequence of retained labels encountered along the hierarchy; it does not contain a simultaneous classical record of every label.

For
\begin{equation}\label{eq:init-en}
|\psi_0\rangle
=
|\psi_0\rangle_D\otimes|0\rangle_A^{\otimes m}
\end{equation}
and $M=I_D\otimes M_A$, the scalar output is
\begin{equation}\label{eq:output-en}
f_{\bm\theta}(x)
=
\langle\psi_0|
U_E(x)^\dagger U_{\bm\theta}^\dagger
M U_{\bm\theta}U_E(x)
|\psi_0\rangle,
\end{equation}
where $U_{\bm\theta}=U_L\cdots U_1$.

\subsection{Operator-Level Implementation}

Let $W_\ell$ extract an atomic sector label into $\mathcal H_T$:
\begin{equation}\label{eq:label-extraction-en}
W_\ell:
|\phi_{\ell,\alpha}\rangle_D|0\rangle_T
\longmapsto
|\phi_{\ell,\alpha}\rangle_D|\alpha\rangle_T.
\end{equation}
The predicate $\alpha\in\mathcal A_{\ell,s}$, made concrete for the $SU(2)$ construction in Eq.~\eqref{eq:retained-from-atoms-en}, controls the corresponding ancilla unitary. Applying $W_\ell^\dagger$ afterwards clears the temporary register and realizes Eq.~\eqref{eq:Ul-block-en} on $\mathcal H_D\otimes\mathcal H_A$.

\subsection{Parameter and Symbolic Circuit-Cost Accounting}

A general traceless Hamiltonian on $\mathcal H_A$ contains at most $d_A^2-1$ real coefficients, giving
\begin{equation}\label{eq:parameter-count-en}
P_{\mathrm{var}}
\le
(d_A^2-1)\sum_{\ell=1}^{L}K_\ell.
\end{equation}
This count excludes the elementary-gate cost of $W_\ell$, its inverse, and the coherent membership predicates. With constant $m$ and $K_\ell$, it scales as $O(\log N)$ along a balanced tree.

To separate variational dimension from implementation overhead, let $C_{\mathrm{ext}}^{(\ell)}$ denote the gate cost of label extraction or uncomputation at layer $\ell$, $C_{\mathrm{pred}}^{(\ell,s)}$ the combined cost of computing and clearing the membership predicate for retained mode $s$, and $C_A^{(\ell,s)}$ the cost of the associated controlled ancilla evolution. A direct extract--control--uncompute realization obeys the symbolic bound
\begin{equation}\label{eq:symbolic-circuit-cost-en}
C_{\mathrm{circ}}
\le
\sum_{\ell=1}^{L}
\left[
2C_{\mathrm{ext}}^{(\ell)}
+
\sum_{s=1}^{K_\ell}
\left(
C_{\mathrm{pred}}^{(\ell,s)}+C_A^{(\ell,s)}
\right)
\right].
\end{equation}
The cost depends on the selected group, sector encoding, and native gate set. Storing a complete layer label requires $\lceil\log_2 R_\ell\rceil$ temporary qubits when $R_\ell$ atomic labels are distinguished. If only a small set of retained modes is queried, their predicates may instead be computed directly and reused, avoiding a full label register. Equation~\eqref{eq:symbolic-circuit-cost-en} leaves both strategies available and makes clear that the logarithmic variational scaling in Eq.~\eqref{eq:parameter-count-en} does not by itself imply logarithmic elementary-gate complexity.

\section{General Theoretical Analysis}

\subsection{Equivariant Evolution and Invariant Readout}

Define $\widetilde\rho(g)=\rho_D(g)\otimes I_A$.

\begin{proposition}\label{prop:equivariance-invariance-en}
If the control projectors satisfy Eq.~\eqref{eq:projector-group-invariant-en}, every controlled layer commutes with $\widetilde\rho(g)$. If the encoding satisfies Eq.~\eqref{eq:encoding-equiv-en}, the data initial state is $G$ invariant, the ancilla carries the trivial representation, and the measured observable acts only on the ancilla, then
\begin{equation}
[U_\ell,\widetilde\rho(g)]=0,
\qquad
f_{\bm\theta}(g\cdot x)=f_{\bm\theta}(x).
\end{equation}
\end{proposition}

\begin{proof}
Each generator term satisfies
\begin{equation}
[\Pi_{\ell,s}\otimes H_{\ell,s},\rho_D(g)\otimes I_A]
=
[\Pi_{\ell,s},\rho_D(g)]\otimes H_{\ell,s}
=0.
\end{equation}
The layer exponentials and their ordered product $U_{\bm\theta}$ commute with $\widetilde\rho(g)$. Substituting Eq.~\eqref{eq:encoding-equiv-en} into Eq.~\eqref{eq:output-en} moves the representation through $U_{\bm\theta}$ and $M$. The remaining operators act on the invariant initial state and cancel.
\end{proof}

The parameterized circuit is equivariant because its evolution intertwines the group action. The measured scalar becomes invariant after imposing the initial-state and readout conditions.

\subsection{Trainable Lie Closure}

The closure analysis uses the complete projective decompositions in Eq.~\eqref{eq:same-layer-orthogonal-en}, their cross-layer commutativity in Eq.~\eqref{eq:cross-commute-en}, and ancilla-local noncommuting generators. The trainable generators and their Lie closure are
\begin{equation}\label{eq:controlled-generators-en}
G_{\ell,s,q}
=
\Pi_{\ell,s}\otimes\Lambda_q,
\qquad q=1,\ldots,d_A^2-1,
\end{equation}
\begin{equation}\label{eq:lie-def-en}
\mathfrak g
=
\Lie\{-\ii G_{\ell,s,q}\}_{\ell,s,q}.
\end{equation}
The complement contributes no generator because its ancilla Hamiltonian is zero.

For $\bm a=(a_1,\ldots,a_L)$ with $a_\ell\in\{1,\ldots,K_\ell,\perp\}$, define
\begin{equation}\label{eq:joint-projector-en}
P_{\bm a}
=
\prod_{\ell=1}^{L}\Pi_{\ell,a_\ell}.
\end{equation}
Equations~\eqref{eq:same-layer-orthogonal-en} and \eqref{eq:cross-commute-en} give
\begin{equation}\label{eq:joint-projector-properties-en}
P_{\bm a}P_{\bm b}
=
\delta_{\bm a\bm b}P_{\bm a},
\qquad
\sum_{\bm a}P_{\bm a}=I_D.
\end{equation}
Let $R=\#\{\bm a:P_{\bm a}\ne0\}$.

\begin{theorem}[Group-independent Lie-algebra containment]\label{thm:lie-containment-en}
Suppose the layer-wise projectors satisfy Eqs.~\eqref{eq:same-layer-orthogonal-en} and \eqref{eq:cross-commute-en}, and all parameter-dependent noncommuting operators are ancilla local as in Eq.~\eqref{eq:controlled-generators-en}. Then
\begin{equation}\label{eq:lie-bound-space-en}
\mathfrak g
\subseteq
\operatorname{span}\{P_{\bm a}:P_{\bm a}\ne0\}
\otimes\su(d_A)
\cong
\bigoplus_{\bm a:P_{\bm a}\ne0}\su(d_A).
\end{equation}
\end{theorem}

\begin{proof}
Completeness gives
\begin{equation}\label{eq:single-to-joint-en}
\Pi_{\ell,s}
=
\sum_{\bm a:a_\ell=s}P_{\bm a}.
\end{equation}
Every generator lies in the right-hand side of Eq.~\eqref{eq:lie-bound-space-en}. This space is closed under commutation since
\begin{equation}
[P_{\bm a}\otimes X,P_{\bm b}\otimes Y]
=
\delta_{\bm a\bm b}P_{\bm a}\otimes[X,Y].
\end{equation}
It therefore contains the generated Lie algebra.
\end{proof}

The theorem is independent of the symmetry group. A chosen group supplies projectors satisfying the stated algebraic conditions.

\begin{corollary}[Dimension upper bound]\label{cor:dimension-en}
\begin{equation}\label{eq:dim-bound-en}
\dim\mathfrak g
\le
R(d_A^2-1)
\le
(d_A^2-1)\prod_{\ell=1}^{L}(K_\ell+1).
\end{equation}
\end{corollary}

\begin{proof}
Each nonzero joint sector supports at most one copy of $\su(d_A)$. There are no more than $\prod_\ell(K_\ell+1)$ label tuples, and incompatible tuples give zero projectors.
\end{proof}

\begin{corollary}[Polynomial growth under truncated control]\label{cor:polynomial-en}
For $d_A=2^m$, $K_\ell\le K$, and $L=\lceil\log_2N\rceil$, constant $m$ and $K$ give
\begin{equation}\label{eq:poly-bound-en}
\dim\mathfrak g
\le
(4^m-1)(K+1)^{\lceil\log_2N\rceil}
=
O\!\left(4^mN^{\log_2(K+1)}\right).
\end{equation}
\end{corollary}

The explicit parameter count in Eq.~\eqref{eq:parameter-count-en} is additive across layers. The Lie closure can resolve the multiplicative collection of compatible joint sectors.

\subsection{Global Joint-Sector Decomposition}

The joint projectors also give an explicit form for the entire variational circuit. For each layer label, define
\begin{equation}\label{eq:path-layer-response-en}
V_{\ell,a_\ell}
=
\begin{cases}
e^{-\ii H_{\ell,a_\ell}(\bm\theta_{\ell,a_\ell})},
&a_\ell\in\{1,\ldots,K_\ell\},\\
I_A,&a_\ell=\perp,
\end{cases}
\end{equation}
and associate a compatible label path $\bm a$ with the ordered ancilla response
\begin{equation}\label{eq:path-response-en}
V_{\bm a}(\bm\theta)
=
V_{L,a_L}\cdots V_{1,a_1}.
\end{equation}

\begin{proposition}[Joint-sector circuit decomposition]\label{prop:global-sector-decomposition-en}
Under Eqs.~\eqref{eq:same-layer-orthogonal-en} and \eqref{eq:cross-commute-en}, the trainable circuit satisfies
\begin{equation}\label{eq:global-sector-decomposition-en}
U_{\bm\theta}
=
\sum_{\bm a:P_{\bm a}\ne0}
P_{\bm a}\otimes V_{\bm a}(\bm\theta).
\end{equation}
For $|\phi_x\rangle_D=U_E(x)|\psi_0\rangle_D$ and an ancilla initial state $|0_A\rangle$, its scalar output can be written as
\begin{equation}\label{eq:sector-mixture-output-en}
f_{\bm\theta}(x)
=
\sum_{\bm a:P_{\bm a}\ne0}
p_{\bm a}(x)\,
\mu_{\bm a}(\bm\theta),
\end{equation}
where
\begin{equation}\label{eq:path-probability-response-en}
\begin{aligned}
p_{\bm a}(x)
&=\langle\phi_x|P_{\bm a}|\phi_x\rangle,\\
\mu_{\bm a}(\bm\theta)
&=\langle0_A|V_{\bm a}^\dagger M_A V_{\bm a}|0_A\rangle.
\end{aligned}
\end{equation}
\end{proposition}

\begin{proof}
Equation~\eqref{eq:single-to-joint-en} inserted into Eq.~\eqref{eq:Ul-block-en} gives
$U_\ell=\sum_{\bm a}P_{\bm a}\otimes V_{\ell,a_\ell}$.
Multiplying the layers and using Eq.~\eqref{eq:joint-projector-properties-en} proves Eq.~\eqref{eq:global-sector-decomposition-en}. Substitution into Eq.~\eqref{eq:output-en} removes cross terms through $P_{\bm a}P_{\bm b}=\delta_{\bm a\bm b}P_{\bm a}$, yielding Eqs.~\eqref{eq:sector-mixture-output-en} and \eqref{eq:path-probability-response-en}.
\end{proof}

The input enters Eq.~\eqref{eq:sector-mixture-output-en} through the symmetry-compatible sector weights $p_{\bm a}(x)$. The trainable circuit assigns an ancilla response to each compatible hierarchical path. This form makes the roles of data-side projective structure and ancilla-side noncommuting dynamics explicit.

\subsection{Expressivity under Sector-Conditioned Dynamics}

Three quantities organize the capacity of the architecture. The number $R$ determines how many compatible paths can contribute distinct responses, $d_A$ sets the dimension of the response system, and the available ancilla generators determine which functions $\mu_{\bm a}(\bm\theta)$ can be produced. Increasing any of these resources can enlarge the model family. Pruning retained modes reduces both the parameter count and the set of controllable paths, providing a direct capacity control.

The path unitaries in Eq.~\eqref{eq:path-response-en} are generally correlated. A layer unitary $V_{\ell,s}$ is reused by every joint sector whose $\ell$th label equals $s$, so the $R$ responses cannot be selected independently with only the additive parameter budget in Eq.~\eqref{eq:parameter-count-en}. This parameter sharing couples related paths and supplies a hierarchical inductive bias. It also explains how the architecture can distinguish many compatible sectors without assigning a separate parameter block to every joint label.

Equation~\eqref{eq:sector-mixture-output-en} identifies a concrete limitation as well. Inputs with identical joint-sector weights cannot be separated by an ancilla-only scalar readout of this form. Finer projectors, additional invariant observables, or a richer encoding can resolve more information when required by the task. The architecture therefore offers a tunable range between aggressive sector aggregation and finer sector-conditioned responses.

\subsection{Trainability Implications}

Theorem~\ref{thm:lie-containment-en} localizes the non-Abelian trainable dynamics to the shared ancilla, while the data-dependent controls remain in the commutative algebra generated by the joint projectors. The accessible unitary orbit consequently avoids arbitrary data-register operators. Related Lie-algebraic quantities constrain quantum Fisher information rank and overparameterization thresholds \cite{larocca_theory_2023,fontana2023adjoint,ragone_lie_2024}.

The ancilla size $m$, retained-mode counts $K_\ell$, and realized path count $R$ regulate the accessible sector-conditioned transformations and algebraic dimension. Empty label tuples lower $R$, and finite-depth generator choices can reduce the realized closure further. Rapid dynamical Lie-algebra growth is one mechanism associated with concentration in randomly initialized variational circuits \cite{mcclean2018barren,cerezo2021cost,holmes_connecting_2022,larocca_review_2024,li_concentration_2022}. Equation~\eqref{eq:dim-bound-en} directly controls this mechanism under the stated structural conditions. It does not constitute a loss-independent guarantee against barren plateaus: gradient statistics also depend on the observable, data distribution, initialization, and depth. The finite-size tests in Sec.~\ref{sec:numerical-en} examine whether the structural restriction is accompanied by a larger initialization signal in the studied setting.

\subsection{Examples beyond \texorpdfstring{$SU(2)$}{SU(2)}}

The construction only requires commuting invariant projectors, so it is not tied to angular-momentum coupling. Consider a $U(1)$ particle-number symmetry on qubits or fermionic modes. For every block $B$ in a fixed nested or disjoint hierarchy, define
\begin{equation}\label{eq:u1-block-number-en}
\mathcal N_B=\sum_{i\in B}n_i,
\qquad
Q_B^{(n)}=\text{the spectral projector of $\mathcal N_B$ at $n$}.
\end{equation}
Block-number operators commute for disjoint blocks. For $C\subset B$, $\mathcal N_B=\mathcal N_C+\mathcal N_{B\setminus C}$ also gives $[\mathcal N_C,\mathcal N_B]=0$. Their spectral projectors therefore admit a common refinement and commute with the global action $e^{-\ii\varphi\mathcal N}$, where $\mathcal N=\sum_i n_i$. Retained number patterns can control the same shared-ancilla dynamics, and Theorem~\ref{thm:lie-containment-en} applies without changing its proof.

A $\mathbb Z_2$ parity hierarchy gives a discrete analogue. Products $Z_B=\prod_{i\in B}Z_i$ on a laminar family of blocks commute, and the projectors $(I\pm Z_B)/2$ provide binary invariant controls whenever the problem symmetry is generated by the corresponding global parity. These examples separate the general algebraic construction from the $SU(2)$ implementation developed next.

\section{\texorpdfstring{$SU(2)$}{SU(2)} Realization with a Fixed Clebsch--Gordan Tree}

\subsection{Rotation Representation and Equivariant Encoding}

Let the input be an ordered collection of three-dimensional vectors
\begin{equation}
x=(\bm p_0,\ldots,\bm p_{N-1}),
\qquad
\bm p_k=(p_{k,x},p_{k,y},p_{k,z})\in\mathbb R^3.
\end{equation}
The coordinates are centered and isotropically rescaled before encoding, so the preprocessing does not select a spatial direction. For $N$ spin-$1/2$ data qubits, a group element $g\in SU(2)$ acts as
\begin{equation}\label{eq:rho-en}
\rho_R(g)=U_g^{\otimes N}.
\end{equation}
The image of $g$ under the double-cover map is denoted by $R_g\in SO(3)$, and its action on the input is $g\cdot x=(R_g\bm p_0,\ldots,R_g\bm p_{N-1})$.

We encode one vector into one data qubit through
\begin{equation}\label{eq:point-encoding-en}
U(\bm p_k)
=
\exp\!\left[
-\frac{\ii}{2}
\left(
p_{k,x}X+p_{k,y}Y+p_{k,z}Z
\right)
\right],
\end{equation}
and take the full encoding to be
\begin{equation}\label{eq:data-encoding-en}
U_E(x)
=
\bigotimes_{k=0}^{N-1}U(\bm p_k).
\end{equation}
Writing $\bm\sigma=(X,Y,Z)$, the defining relation between the two rotation representations is
\begin{equation}\label{eq:pauli-vector-covariance-en}
U_g(\bm p\cdot\bm\sigma)U_g^\dagger
=
(R_g\bm p)\cdot\bm\sigma.
\end{equation}
Unitary conjugation commutes with the matrix exponential. Equations~\eqref{eq:point-encoding-en} and \eqref{eq:pauli-vector-covariance-en} therefore give
\begin{equation}
U(R_g\bm p_k)
=
U_gU(\bm p_k)U_g^\dagger.
\end{equation}
Taking the tensor product over all points establishes the concrete encoding relation
\begin{equation}\label{eq:su2-encoding-en}
U_E(g\cdot x)
=
\rho_R(g)U_E(x)\rho_R(g)^\dagger.
\end{equation}

\subsection{Invariant Initial State and Its Preparation}

For even $N$, pair the data qubits according to the bottom layer of the coupling tree. Each pair is initialized in the singlet
\begin{equation}\label{eq:singlet-state-en}
|s\rangle
=
\frac{|01\rangle-|10\rangle}{\sqrt2},
\end{equation}
and the data-register initial state is
\begin{equation}\label{eq:singlet-product-state-en}
|\psi_0\rangle_D
=
\bigotimes_{k=0}^{N/2-1}|s\rangle_{2k,2k+1}.
\end{equation}
The singlet is the antisymmetric tensor in two dimensions. Hence
\begin{equation}
(U_g\otimes U_g)|s\rangle
=
\det(U_g)|s\rangle
=
|s\rangle
\end{equation}
for every $U_g\in SU(2)$, which gives
\begin{equation}\label{eq:initial-state-invariance-en}
\rho_R(g)|\psi_0\rangle_D
=
|\psi_0\rangle_D.
\end{equation}

\begin{figure}[t]
  \centering
  \begin{quantikz}[row sep=0.45cm, column sep=0.5cm]
    \lstick{$|0\rangle$} & \gate{H} & \gate{Z} & \ctrl{1} & \qw \\
    \lstick{$|0\rangle$} & \gate{X} & \qw      & \targ{}  & \qw
  \end{quantikz}
  \caption{Constant-depth preparation of the two-qubit singlet used in each bottom-layer pair.}
  \label{fig:singlet-preparation-en}
\end{figure}
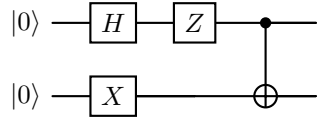

Figure~\ref{fig:singlet-preparation-en} prepares Eq.~\eqref{eq:singlet-state-en} from $|00\rangle$. Applying this circuit in parallel to $(0,1),(2,3),\ldots$ matches the two-qubit blocks at the bottom of the fixed Clebsch--Gordan tree of Fig.~\ref{fig:hierarchy-en}. The shared ancilla is initialized as $|0\rangle_A^{\otimes m}$ and carries the trivial representation. Defining
\begin{equation}
|\phi(x)\rangle_D
=
U_E(x)|\psi_0\rangle_D,
\end{equation}
Eqs.~\eqref{eq:su2-encoding-en} and \eqref{eq:initial-state-invariance-en} yield the state-level encoding covariance
\begin{equation}\label{eq:encoded-state-equivariance-en}
|\phi(g\cdot x)\rangle_D
=
\rho_R(g)|\phi(x)\rangle_D.
\end{equation}

\subsection{Recursive Spin-Sector Projectors}

A fixed binary Clebsch--Gordan tree recursively couples adjacent spins. For $N=2^L$, layer $\ell$ contains
\begin{equation}\label{eq:block-en}
B_{\ell,r}
=
\{2^\ell r,\ldots,2^\ell(r+1)-1\},
\qquad
r=0,\ldots,\frac{N}{2^\ell}-1.
\end{equation}
The same merging rule produces depth $O(\log N)$ for other even $N$.

\begin{figure}[t]
  \centering
  \includegraphics[width=\columnwidth]{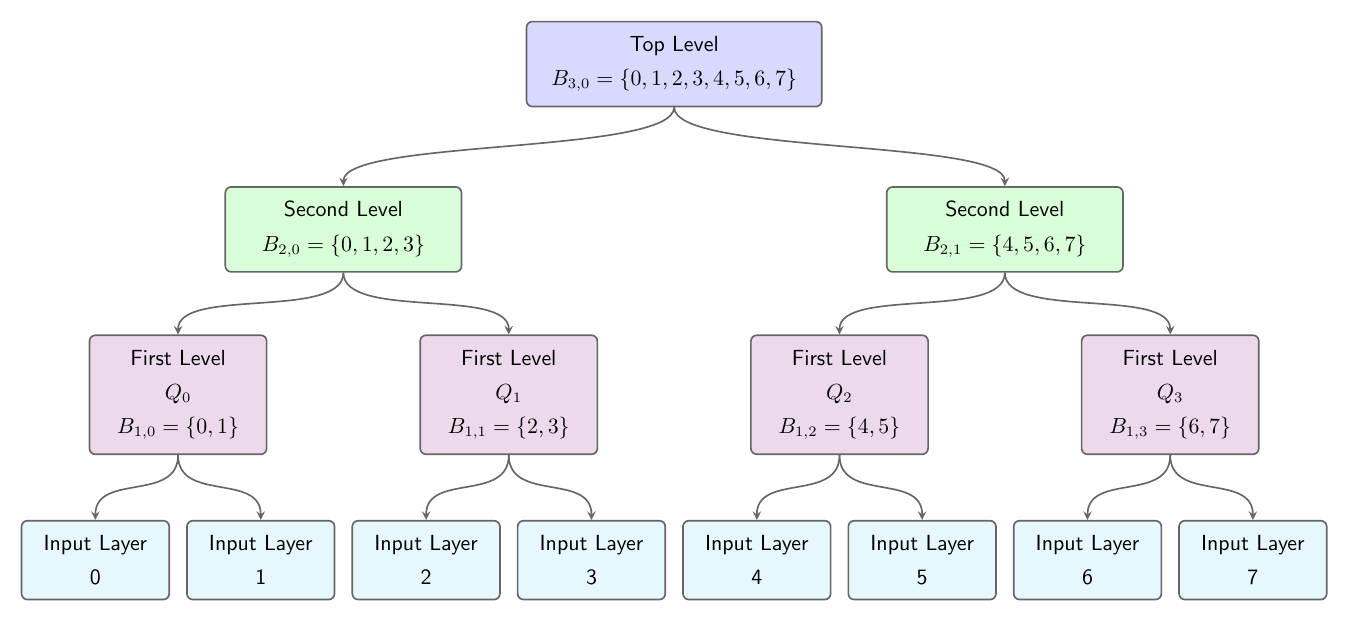}
  \caption{Balanced hierarchical merging along a fixed $SU(2)$ spin-coupling tree.}
  \label{fig:hierarchy-en}
\end{figure}

For every internal block $B$, define
\begin{equation}\label{eq:intermediate-spin-en}
\mathbf J_B^2
=
\left(\sum_{i\in B}\mathbf S_i\right)^2,
\qquad
\mathbf S_i=\frac12(X_i,Y_i,Z_i).
\end{equation}

\begin{lemma}\label{lem:tree-projectors-commute-en}
The intermediate-spin operators of a fixed coupling tree commute. Their joint spectral projectors form a common refinement and commute with $\rho_R(g)$.
\end{lemma}

\begin{proof}
Two tree blocks are disjoint or nested. Operators on disjoint blocks commute. For $C\subset B$, write $\mathbf J_B=\mathbf J_C+\mathbf J_{B\setminus C}$. The Casimir $\mathbf J_C^2$ commutes with each component of $\mathbf J_C$, and the remaining components act outside $C$. Hence $[\mathbf J_C^2,\mathbf J_B^2]=0$. Each intermediate Casimir is invariant under simultaneous spin rotations, so its spectral projectors commute with the global representation.
\end{proof}

For the two-spin decomposition $\frac12\otimes\frac12=0\oplus1$, define
\begin{equation}
C=X\otimes X+Y\otimes Y+Z\otimes Z.
\end{equation}
The total angular momentum satisfies
\begin{equation}
\mathbf J^2
=(\mathbf S_1+\mathbf S_2)^2
=\frac32I+\frac12C.
\end{equation}
Since $\mathbf J^2$ has eigenvalues $j(j+1)$ for $j=0,1$, the corresponding eigenvalues of $C$ are $-3$ and $1$. Spectral interpolation gives
\begin{equation}\label{eq:singlet-triplet-projectors-en}
P^{(0)}=\frac14(I-C),
\qquad
P^{(1)}=\frac14(3I+C).
\end{equation}
These orthogonal projectors depend only on total spin and commute with $U_g\otimes U_g$.

Let $\mathcal T_{\mathrm{int}}$ denote the internal blocks of the fixed binary coupling tree. If an internal block $B$ has children $B_L$ and $B_R$ carrying spins $j_{B_L}$ and $j_{B_R}$, angular-momentum addition permits
\begin{equation}\label{eq:recursive-spin-range-en}
j_B
\in
\left\{
|j_{B_L}-j_{B_R}|,
|j_{B_L}-j_{B_R}|+1,
\ldots,
j_{B_L}+j_{B_R}
\right\}.
\end{equation}
A complete fixed-tree label is
\begin{equation}\label{eq:tree-label-en}
\alpha=\{j_B:B\in\mathcal T_{\mathrm{int}}\}.
\end{equation}
Writing $P_B^{(j_B)}$ for the spectral projector of $\mathbf J_B^2$ with eigenvalue $j_B(j_B+1)$, define the atomic projector
\begin{equation}\label{eq:recursive-atomic-projector-en}
P_\alpha
=
\prod_{B\in\mathcal T_{\mathrm{int}}}P_B^{(j_B)}.
\end{equation}
Lemma~\ref{lem:tree-projectors-commute-en} makes the product order irrelevant. A label violating Eq.~\eqref{eq:recursive-spin-range-en} at any merge has $P_\alpha=0$; the compatible labels give mutually orthogonal projectors whose sum is $I_D$.

Joint eigenvalues of the intermediate Casimirs label atomic projectors $\{P_\alpha\}$. A retained mode is specified by the spins on one tree layer. For layer $\ell$, let $\mathcal B_\ell=\{B_{\ell,r}\}_r$ collect the blocks of Eq.~\eqref{eq:block-en} and write $\alpha|_{\mathcal B_\ell}$ for the restriction of a complete label to these blocks. Retained mode $s$ fixes a joint spin pattern $\bm\jmath_{\ell,s}$ on $\mathcal B_\ell$ and gathers the atoms compatible with it:
\begin{equation}\label{eq:retained-from-atoms-en}
\mathcal A_{\ell,s}
=
\bigl\{\alpha:\alpha|_{\mathcal B_\ell}=\bm\jmath_{\ell,s}\bigr\},
\qquad
\Pi_{\ell,s}
=
\sum_{\alpha\in\mathcal A_{\ell,s}}P_\alpha,
\qquad
\mathcal A_{\ell,s}\cap\mathcal A_{\ell,t}=\varnothing
\quad(s\ne t).
\end{equation}
The patterns $\bm\jmath_{\ell,1},\ldots,\bm\jmath_{\ell,K_\ell}$ are distinct, which gives the disjointness, and labels outside $\bigcup_s\mathcal A_{\ell,s}$ form the complement mode $\Pi_{\ell,\perp}$ of Eq.~\eqref{eq:complement-en}. The disjointness is imposed within a single layer; the same atom may enter retained modes at several layers.
These projectors supply the concrete controls of the parameterized circuit and its equivariance condition, since
\begin{equation}\label{eq:su2-retained-projector-invariance-en}
[\Pi_{\ell,s},\rho_R(g)]=0.
\end{equation}
They also provide the common refinement required by Eq.~\eqref{eq:cross-commute-en}. Changing coupling orders across layers would generally remove this property.

For $N=4$, take the tree $((0,1),(2,3))$. Its layer-one patterns are
\begin{equation}\label{eq:n4-layer-one-atoms-en}
P_{01}^{(j_{01})}P_{23}^{(j_{23})},
\qquad
j_{01},j_{23}\in\{0,1\}.
\end{equation}
An illustrative truncation is
\begin{equation}\label{eq:n4-layer-one-retained-en}
\begin{aligned}
\Pi_{1,1}&=P_{01}^{(0)}P_{23}^{(0)},\\
\Pi_{1,2}&=P_{01}^{(1)}P_{23}^{(1)},\\
\Pi_{1,\perp}
&=P_{01}^{(0)}P_{23}^{(1)}
 +P_{01}^{(1)}P_{23}^{(0)}.
\end{aligned}
\end{equation}
Here $K_1=2$ counts two layer-wise joint patterns.

Let $P_{0123}^{(J)}$ be the four-qubit total-spin projector. The fixed-tree atoms are
\begin{equation}\label{eq:n4-joint-atomic-en}
P_{j_{01},j_{23},J}
=
P_{01}^{(j_{01})}P_{23}^{(j_{23})}P_{0123}^{(J)},
\end{equation}
with incompatible triples giving zero. They resolve
\begin{equation}
P_{0123}^{(J)}
=
\sum_{j_{01},j_{23}}P_{j_{01},j_{23},J}.
\end{equation}
Retaining $J=0$ and $J=1$ at layer two gives
\begin{equation}\label{eq:n4-layer-two-retained-en}
\Pi_{2,1}=P_{0123}^{(0)},
\qquad
\Pi_{2,2}=P_{0123}^{(1)},
\qquad
\Pi_{2,\perp}=P_{0123}^{(2)}.
\end{equation}
This $K_1=K_2=2$ example displays layer-wise counting, the complement, and cross-layer compatibility.

\subsection{Hierarchical Equivariant Parameterized Circuit}

The spin-sector projectors instantiate the shared-ancilla response of the general architecture. In the $SU(2)$ realization, layer $\ell$ is
\begin{equation}\label{eq:su2-controlled-layer-en}
U_\ell(\bm\theta_\ell)
=
\exp\!\left[
-\ii\sum_{s=1}^{K_\ell}
\Pi_{\ell,s}\otimes
H_{\ell,s}(\bm\theta_{\ell,s})
\right].
\end{equation}
Define $\widetilde\rho_R(g)=\rho_R(g)\otimes I_A$. Equation~\eqref{eq:su2-retained-projector-invariance-en} gives
\begin{equation}
\left[
\sum_{s=1}^{K_\ell}
\Pi_{\ell,s}\otimes H_{\ell,s},
\widetilde\rho_R(g)
\right]
=
\sum_{s=1}^{K_\ell}
[\Pi_{\ell,s},\rho_R(g)]\otimes H_{\ell,s}
=0.
\end{equation}
Consequently,
\begin{equation}\label{eq:su2-layer-equivariance-en}
[U_\ell(\bm\theta_\ell),\widetilde\rho_R(g)]=0,
\qquad
[U_{\bm\theta},\widetilde\rho_R(g)]=0,
\end{equation}
where $U_{\bm\theta}=U_L\cdots U_1$.

For an operator-level realization, illustrated in Fig.~\ref{fig:architecture-en}, the fixed-tree spin labels are coherently extracted into the temporary register. Membership in $\mathcal A_{\ell,s}$ controls the corresponding ancilla evolution, and the extraction is then reversed. The extraction unitary is built from invariant projectors, so it commutes with $\rho_R(g)\otimes I_T$. The temporary register returns to its reference state after each controlled operation and is omitted from Eq.~\eqref{eq:su2-controlled-layer-en}.

\begin{figure}[t]
  \centering
  \includegraphics[width=1\textwidth]{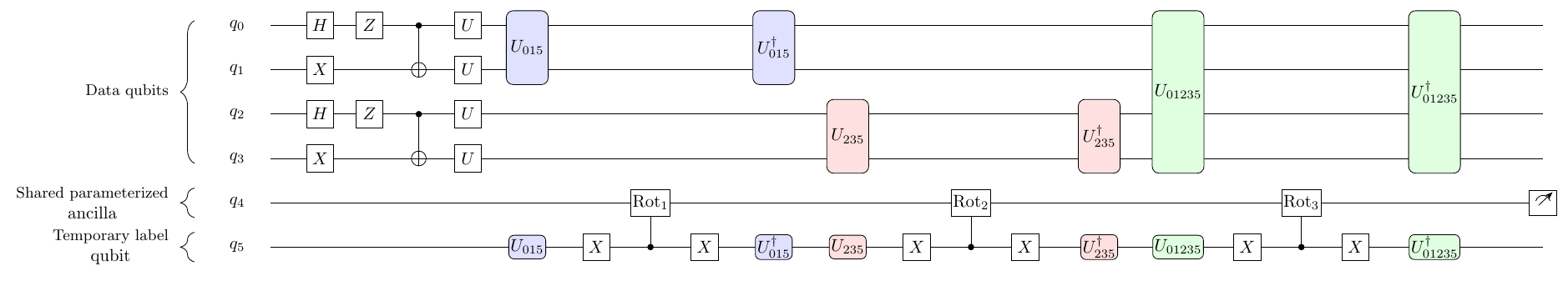}
  \caption{The $N=4$ $SU(2)$ realization. The equivariantly encoded data register supplies hierarchical spin-sector labels, the shared ancilla carries the parameterized response, and the temporary label register is uncomputed after each controlled operation.}
  \label{fig:architecture-en}
\end{figure}

Equation~\eqref{eq:su2-layer-equivariance-en} shows that the hierarchical controlled evolution is an equivariant parameterized circuit. Its noncommuting trainable response remains confined to the shared ancilla, as characterized by Theorem~\ref{thm:lie-containment-en}.

\subsection{Ancilla Readout and End-to-End Invariance}

After every temporary label register has been uncomputed, define the complete output state by
\begin{equation}\label{eq:complete-output-state-en}
|\Psi_{\bm\theta}(x)\rangle
=
U_{\bm\theta}
(U_E(x)\otimes I_A)
\left(
|\psi_0\rangle_D\otimes|0\rangle_A^{\otimes m}
\right).
\end{equation}
The ancilla carries the trivial representation. Combining Eqs.~\eqref{eq:su2-encoding-en}, \eqref{eq:initial-state-invariance-en}, and \eqref{eq:su2-layer-equivariance-en} gives
\begin{equation}\label{eq:complete-state-equivariance-en}
|\Psi_{\bm\theta}(g\cdot x)\rangle
=
\widetilde\rho_R(g)
|\Psi_{\bm\theta}(x)\rangle.
\end{equation}
Thus the complete encoded and parameterized state transforms equivariantly under a global rotation.

Let $M_A=M_A^\dagger$ be a general observable on the shared ancilla and set
\begin{equation}\label{eq:su2-ancilla-measurement-en}
M=I_D\otimes M_A.
\end{equation}
No spatial representation acts on $\mathcal H_A$, and hence
\begin{equation}
[M,\widetilde\rho_R(g)]=0.
\end{equation}
The scalar model output is
\begin{equation}\label{eq:su2-output-en}
f_{\bm\theta}(x)
=
\langle\Psi_{\bm\theta}(x)|
M
|\Psi_{\bm\theta}(x)\rangle.
\end{equation}
Using Eq.~\eqref{eq:complete-state-equivariance-en},
\begin{equation}
\begin{aligned}
f_{\bm\theta}(g\cdot x)
&=
\langle\Psi_{\bm\theta}(x)|
\widetilde\rho_R(g)^\dagger
M
\widetilde\rho_R(g)
|\Psi_{\bm\theta}(x)\rangle\\
&=f_{\bm\theta}(x).
\end{aligned}
\end{equation}
The encoding and hierarchical parameterized evolution therefore define a rotation-equivariant quantum state map. The ancilla-side scalar measurement converts that state map into a rotation-invariant model output.
\section{Numerical Simulations}
\label{sec:numerical-en}

The simulations examine three finite-size consequences of the construction: initialization-gradient statistics relative to generic and conventional equivariant circuits, the contribution of hierarchical sector control relative to a global ancilla baseline, and the capacity to fit representative invariant functions. Every experiment instantiates the general architecture with $SU(2)$ projectors and uses exact state-vector simulation. Hardware noise and finite-shot sampling are absent. Both effects can alter equivariance and trainability diagnostics \cite{tuysuz_symmetry_2024,yao_direct_2025}, so the comparisons below concern circuit structure under ideal simulation.

\subsection{Gradient Variance and Gradient Norm at Initialization}

We first examine the initialization-gradient scaling suggested by the Lie-algebraic restriction. Three variational families are compared: a generic parameterized circuit, a conventional rotationally equivariant circuit, and the proposed hierarchical ancilla-controlled architecture instantiated with $SU(2)$ sector projectors. The number of data qubits is chosen from
\[
N\in\{4,6,8,10,12\}.
\]
For each system size, we perform $30$ independent random parameter initializations and data samplings. All trainable parameters are initialized from the uniform distribution $\mathcal U(0,\pi)$. Let $\mathcal L(\bm\theta)$ be the loss function for one initialization and let
\begin{equation}
\nabla_{\bm\theta}\mathcal L
=
\left(
\frac{\partial\mathcal L}{\partial\theta_1},
\frac{\partial\mathcal L}{\partial\theta_2},
\ldots,
\frac{\partial\mathcal L}{\partial\theta_P}
\right)
\end{equation}
be the parameter-gradient vector, where $P$ is the number of trainable parameters. We evaluate two quantities. The first is the gradient variance over trainable parameters, which reflects whether individual parameter directions vanish as the system size increases. The second is the mean squared gradient norm,
\begin{equation}
\|\nabla_{\bm\theta}\mathcal L\|_2^2
=
\sum_{p=1}^{P}
\left(
\frac{\partial\mathcal L}{\partial\theta_p}
\right)^2,
\end{equation}
which measures the total optimizable signal across all parameter directions. If both quantities decay rapidly with $N$, the local loss landscape around random initialization becomes nearly flat, and gradient-based optimizers require increasingly precise gradient estimates.

Figure~\ref{fig:gradient-summary-en}(a) shows the gradient variance as a function of system size. The generic and conventional rotationally equivariant circuits both decay rapidly as $N$ increases, with the latter retaining a larger scale. The hierarchical controlled architecture exhibits slower decay over $N\in\{4,6,8,10,12\}$. Its mean squared gradient norm is also larger across the same range [Fig.~\ref{fig:gradient-summary-en}(b)]. These finite-size statistics agree with the gradient behavior expected from the restricted trainable algebra.

\begin{figure}[t]
  \centering
  \begin{minipage}{0.32\textwidth}
    \centering
    \includegraphics[width=\linewidth]{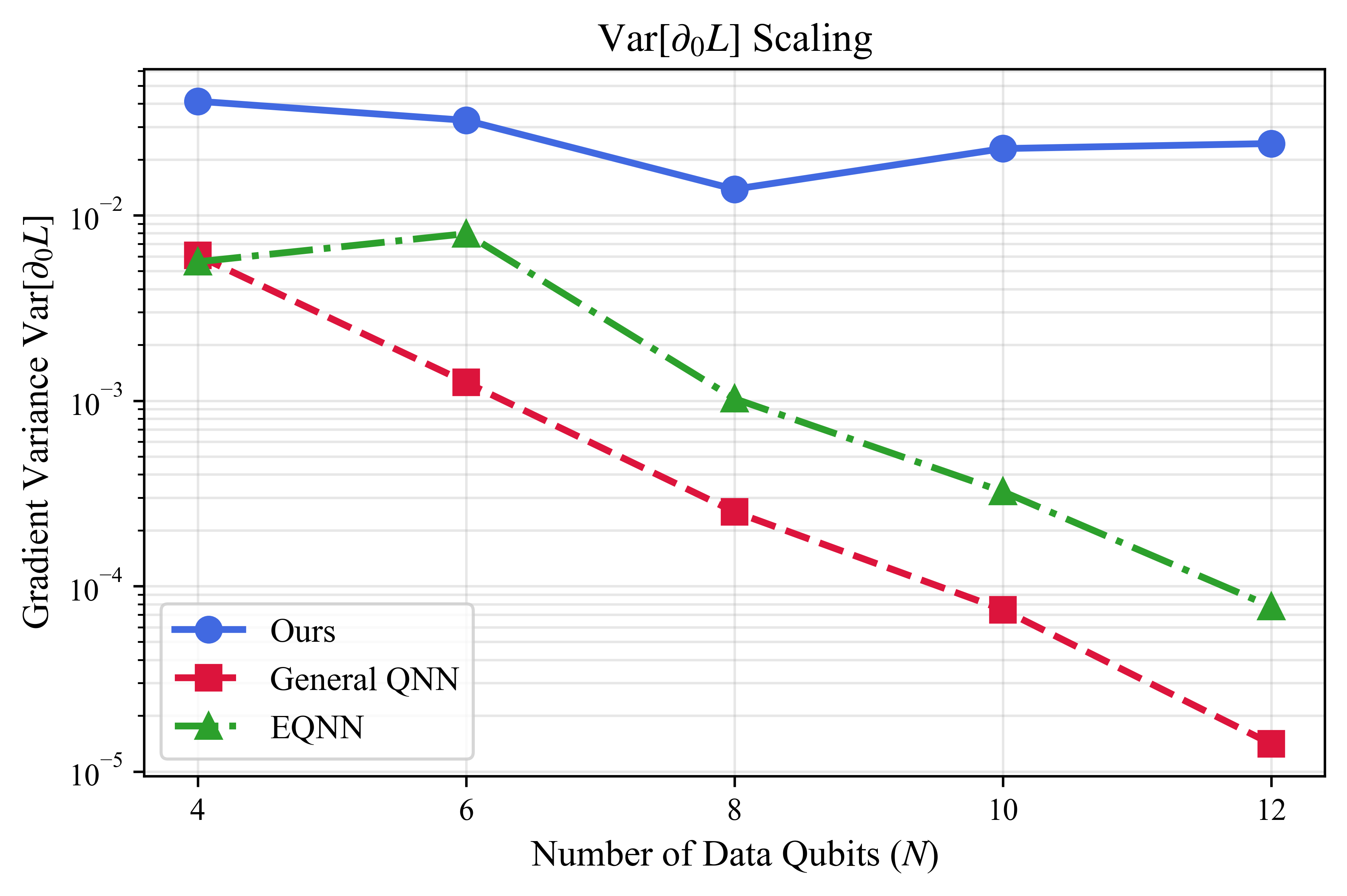}\\
    (a)
  \end{minipage}
  \begin{minipage}{0.32\textwidth}
    \centering
    \includegraphics[width=\linewidth]{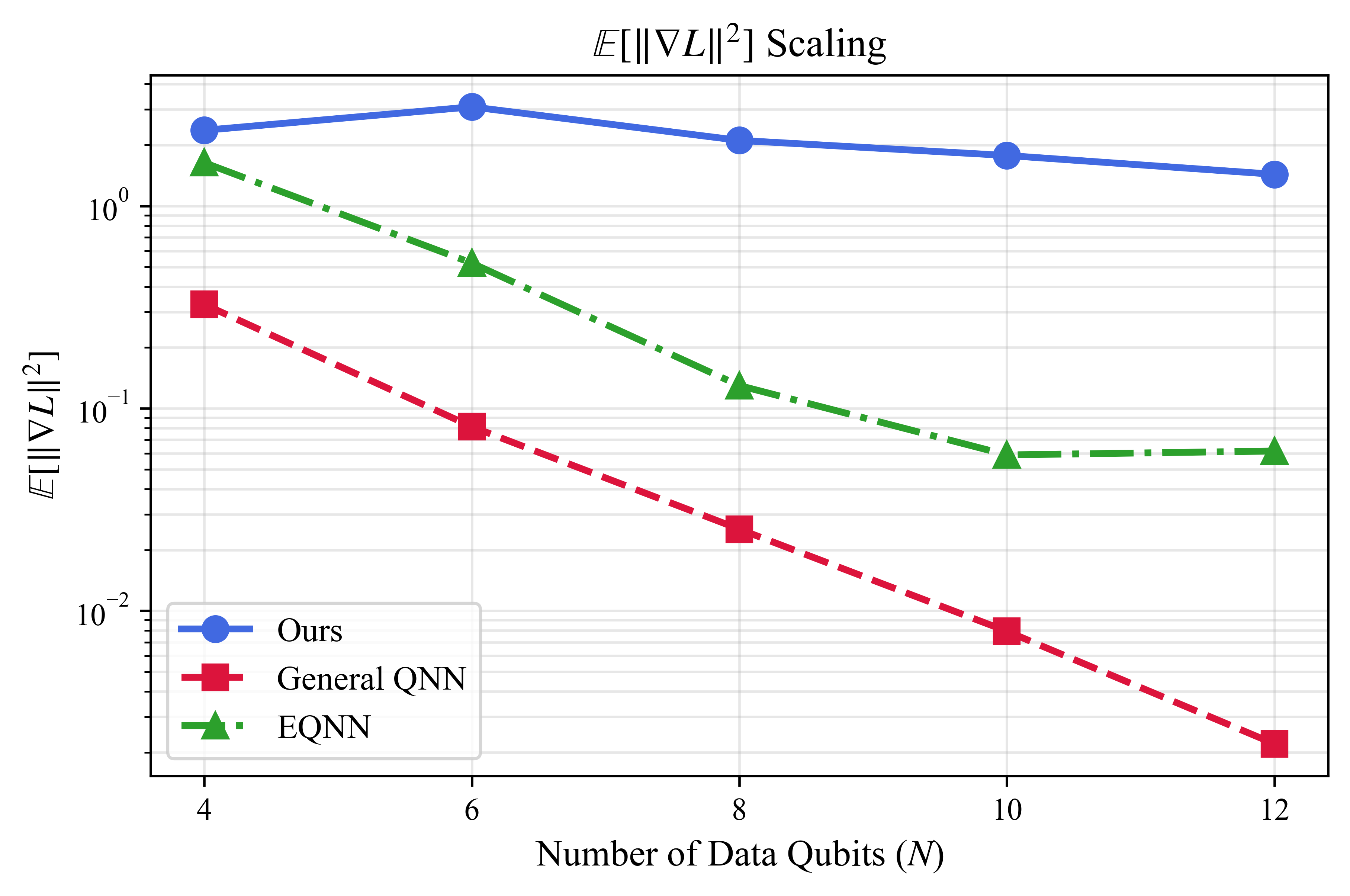}\\
    (b)
  \end{minipage}
  \begin{minipage}{0.32\textwidth}
    \centering
    \includegraphics[width=\linewidth]{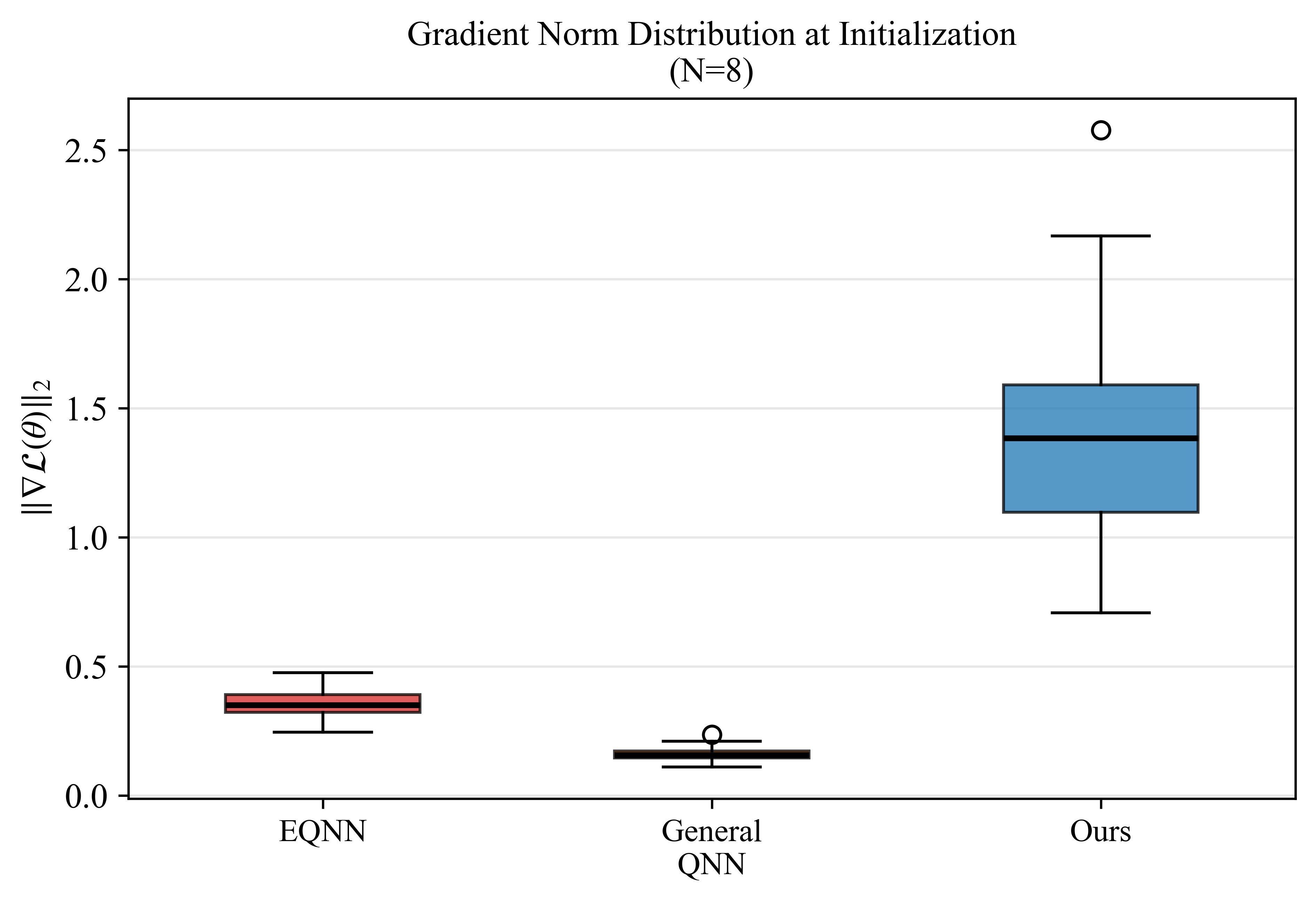}\\
    (c)
  \end{minipage}
  \caption{Initialization-gradient diagnostics. (a) Gradient variance as a function of data-qubit number. (b) Mean squared gradient norm. (c) Gradient-norm distribution at $N=8$.}
  \label{fig:gradient-summary-en}
\end{figure}

To inspect the initialization distribution, we fix $N=8$, perform $60$ random initializations for each model, and record the $L_2$ gradient norm. The box plot in Fig.~\ref{fig:gradient-summary-en}(c) separates a systematic shift from isolated large-gradient outliers. The hierarchical architecture has a higher central range, whereas the generic circuit places more samples near zero. The distribution supports the same finite-size trend as the variance and mean-squared-norm results.

\subsection{Ancilla Resource Ablation}

The ancilla ablation separates register size from hierarchical projection control for $m=1$ and $m=2$. Its baseline adds the same number of ancilla qubits to a generic parameterized model and uses global data-register control without the coupling-tree projectors.

Figure~\ref{fig:ancilla-ablation-en} shows that increasing $m$ lowers the gradient variance and mean squared norm in both models. Ancilla enlargement alone therefore does not account for the observed gradient scale. At each tested $m$, hierarchical projection control retains larger values than global control. The comparison isolates the $SU(2)$-structured control conditions as the distinguishing circuit feature in this ablation.

\begin{figure}[t]
  \centering
  \begin{minipage}{0.48\textwidth}
    \centering
    \includegraphics[width=\linewidth]{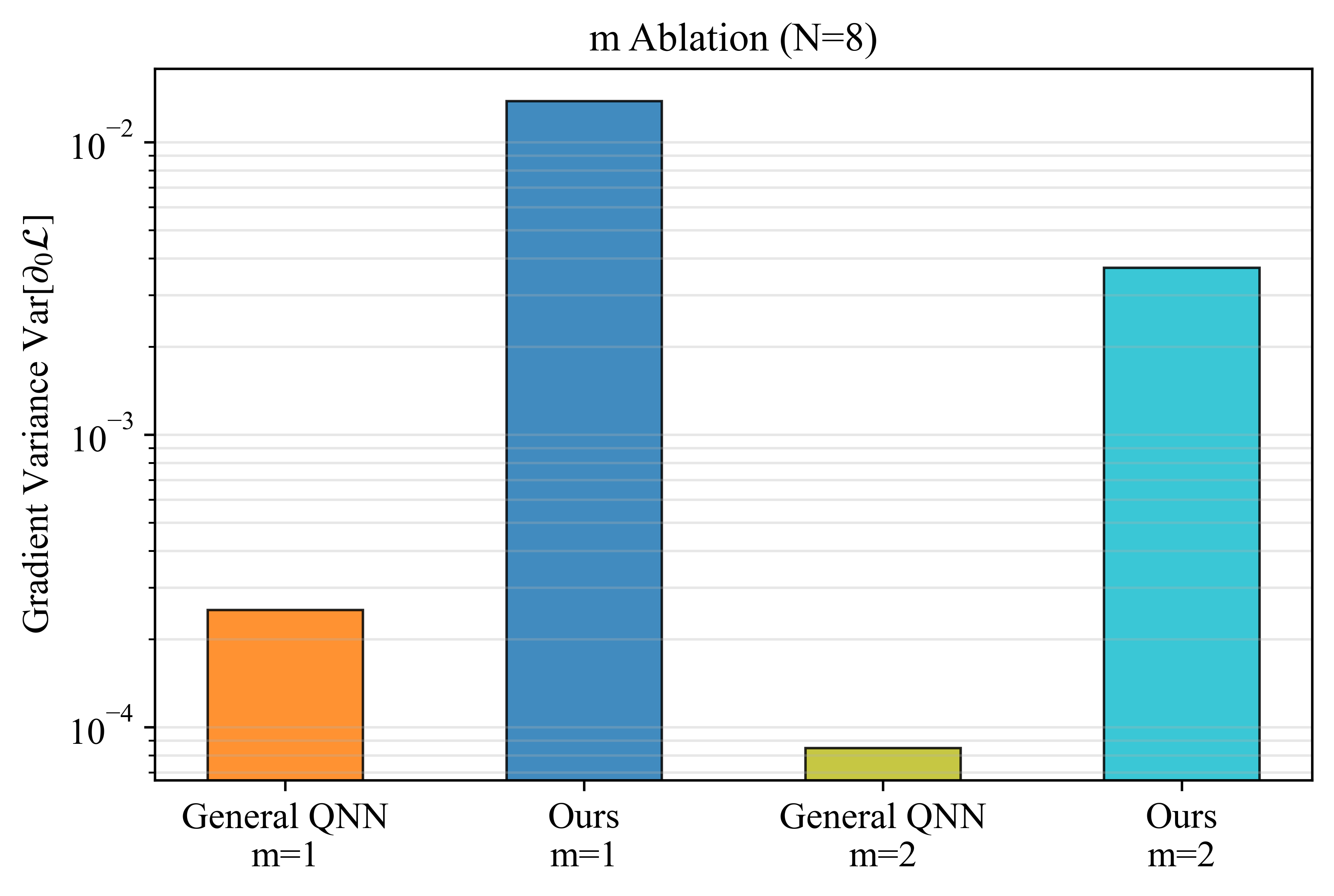}\\
    (a)
  \end{minipage}
  \begin{minipage}{0.48\textwidth}
    \centering
    \includegraphics[width=\linewidth]{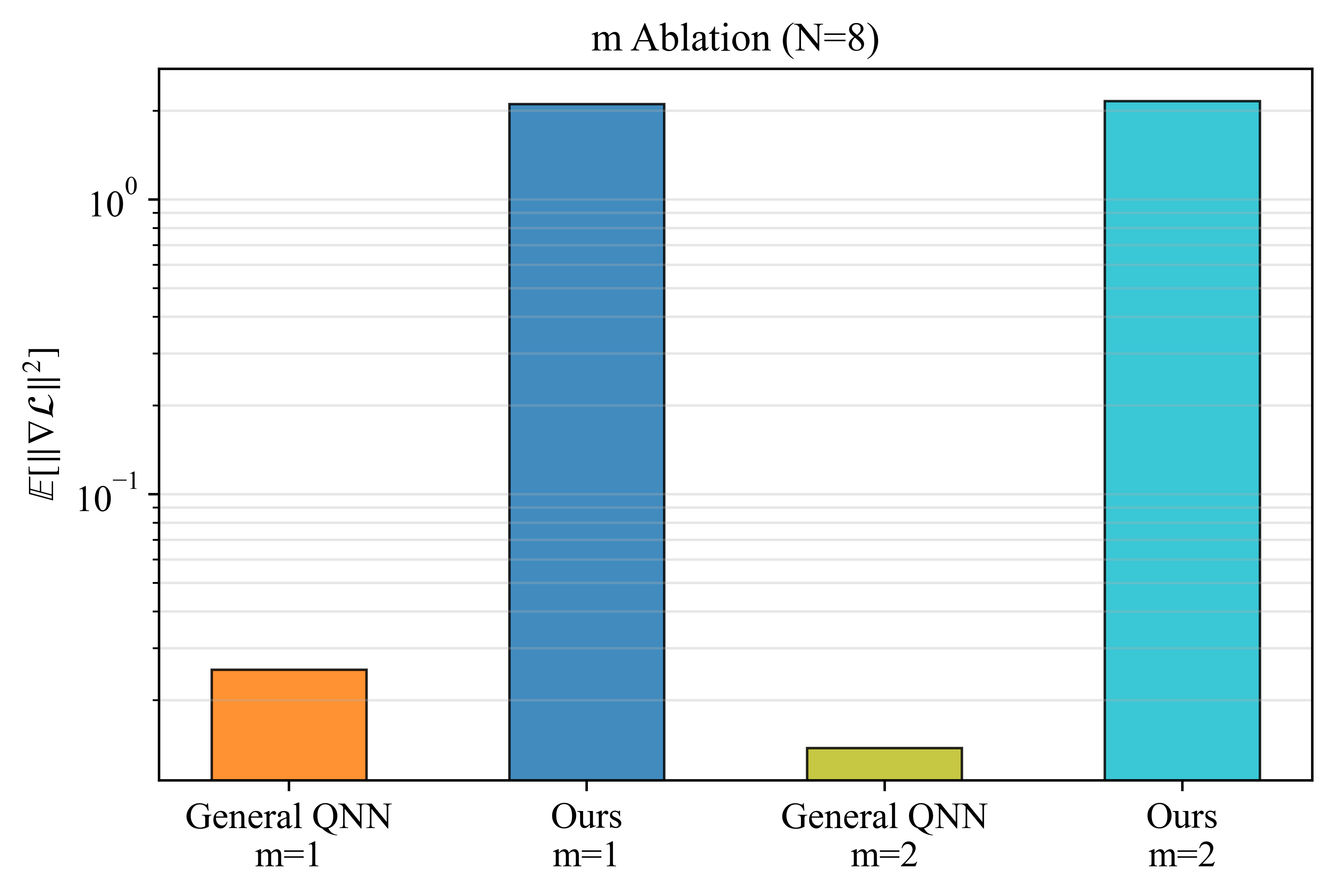}\\
    (b)
  \end{minipage}
  \caption{Ancilla-resource ablation. (a) Gradient variance for $m=1$ and $m=2$. (b) Mean squared gradient norm for $m=1$ and $m=2$.}
  \label{fig:ancilla-ablation-en}
\end{figure}

\subsection{Rotation-Invariant Learning Benchmarks}

The learning benchmarks examine whether projection pruning and ancilla restriction leave enough capacity for representative rotation-invariant maps. They provide finite-size tests of the $SU(2)$ realization on geometric and physical targets.

We first test it on a binary classification task using sampled sphere and torus point clouds. The model uses $N=4$ data qubits, $m=1$ shared ancilla qubit, and circuit depth $5$. The optimizer is Adam with learning rate $10^{-3}$. The output is an ancilla expectation value and the training loss is mean squared error. If labels are normalized to $\{-1,1\}$ and the model output satisfies $f_{\bm\theta}(x)\in[-1,1]$, the loss is
\begin{equation}
\mathcal L_{\mathrm{cls}}(\bm\theta)
=
\frac{1}{B}\sum_{b=1}^{B}
\left(f_{\bm\theta}(x_b)-y_b\right)^2,
\end{equation}
where $B$ is the batch size and $y_b$ is the binary label. The predicted class is determined by the sign of the output. Results averaged over $5$ random seeds show that the test accuracy quickly exceeds $0.9$ in the early training stage and stabilizes around $0.97$.

\begin{figure}[t]
  \centering
  \includegraphics[width=0.86\textwidth]{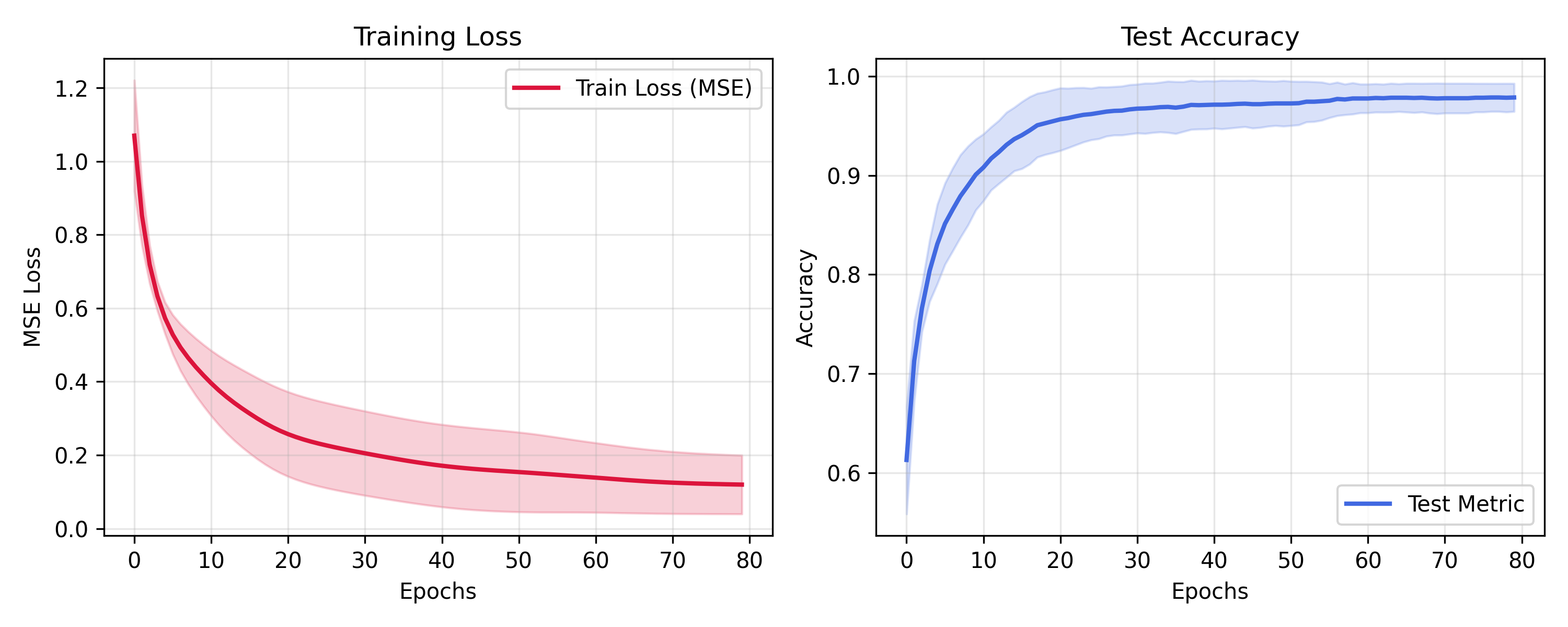}
  \caption{Training loss and test accuracy on the sphere-torus binary classification task.}
  \label{fig:sphere-torus-en}
\end{figure}

We further evaluate the method on ModelNet-5, a sparse point-cloud classification benchmark constructed from the official ModelNet40 HDF5 data \cite{wu_3dshapenets_2015}. Point-cloud learning is a standard setting where geometric priors and rotation-invariant or rotation-equivariant representations are important \cite{bronstein_geometric_2017,qi_pointnet_2017,sebastian_image_2025}. The dataset contains five classes: bottle, bowl, cup, lamp, and stool, corresponding to the original class indices $[0,8,17,20,35]$. Each object originally contains $2048$ points. To fit the current quantum simulation scale, we apply farthest point sampling and retain $4$ points per object. The training, validation, and test sets contain $3500$, $500$, and $1000$ samples, respectively.

The five-class task uses $5$ ancilla readouts as a feature vector followed by a two-layer MLP classifier. Here $m=5$ is fixed by the readout dimension and remains independent of the data-register size, consistent with the constant-$m$ scaling regime. The quantum circuit contains $380$ parameters and the classical head contains $181$, for $561$ parameters in total. A hardware-efficient ansatz (HEA) uses the same numbers of data and ancilla qubits with comparable quantum and classical parameter counts. This matching emphasizes the effect of circuit structure.

Figure~\ref{fig:modelnet-curve-en} shows the training curves over $3$ random seeds. The proposed architecture decreases the loss faster during early training and reaches a validation accuracy near $0.75$; the HEA baseline remains near $0.5$. On the test set [Fig.~\ref{fig:modelnet-test-en}], the respective accuracies are $75.6\pm0.8\%$ and $50.5\pm1.2\%$. Under this sparse four-point input and matched parameter scale, invariant-sector control is associated with faster optimization and higher classification accuracy.

\begin{figure}[t]
  \centering
  \includegraphics[width=0.86\textwidth]{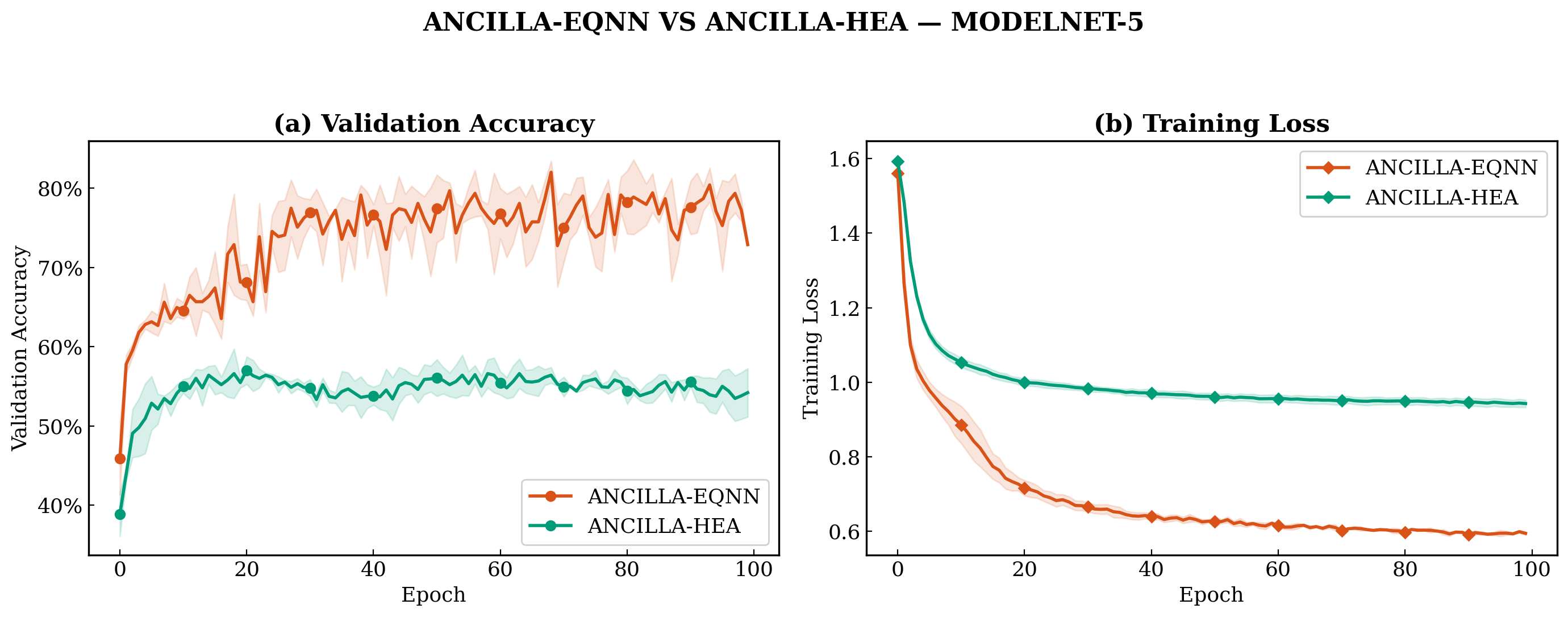}
  \caption{Training loss and classification accuracy on the ModelNet-5 dataset.}
  \label{fig:modelnet-curve-en}
\end{figure}

\begin{figure}[t]
  \centering
  \includegraphics[width=0.55\textwidth]{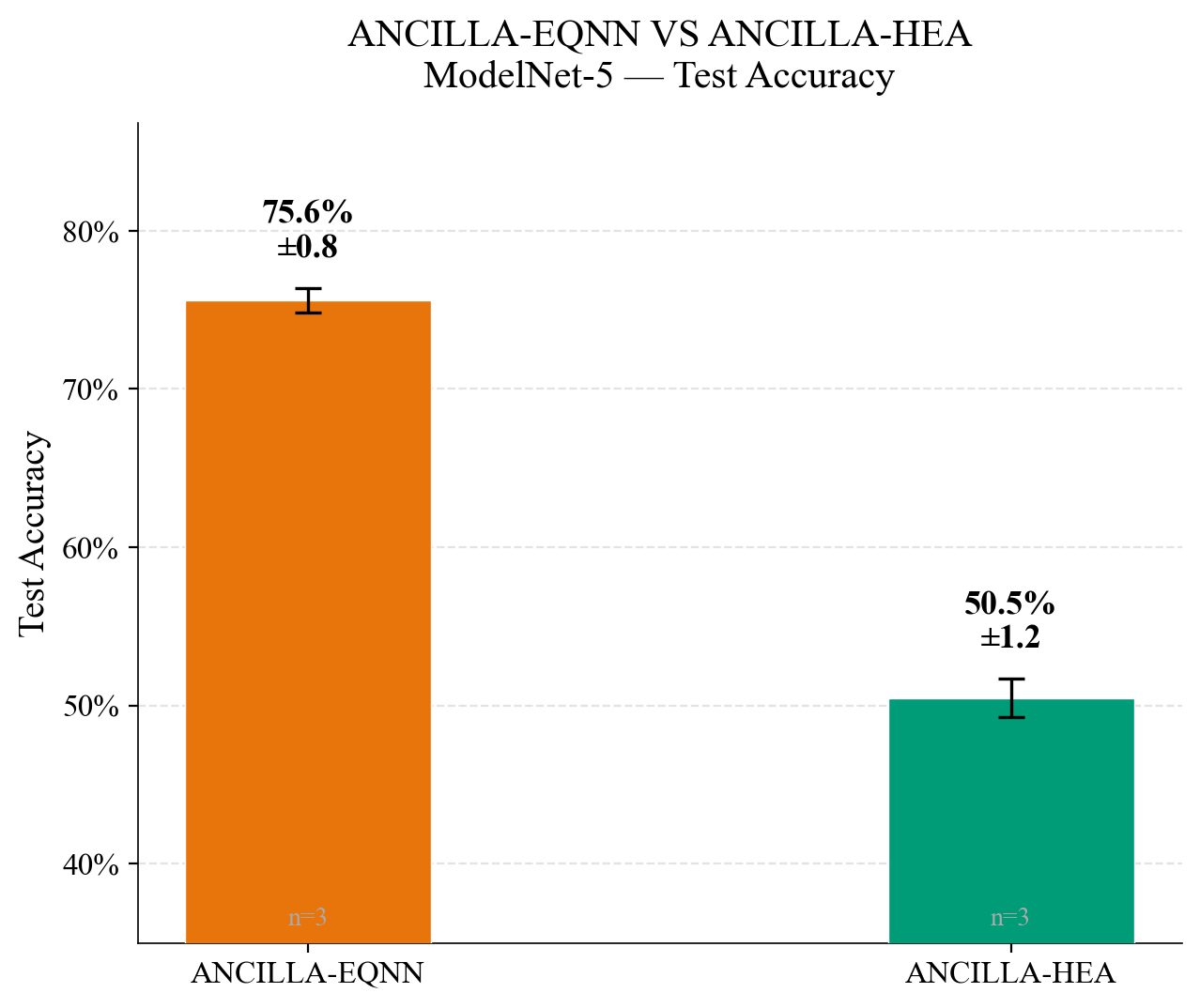}
  \caption{Test accuracy comparison on the ModelNet-5 dataset.}
  \label{fig:modelnet-test-en}
\end{figure}

\subsection{Heisenberg Ground-State Energy Regression}

Because the $SU(2)$ realization directly uses invariant subspaces and spin-coupling information, we next consider a physics-motivated regression problem. The task is to predict the ground-state energy of a geometry-dependent isotropic Heisenberg model. Each sample consists of $4$ spin positions randomly distributed in the unit ball. The coupling strength between sites $i$ and $j$ is
\begin{equation}
J_{ij}=\exp\left(-\lVert\bm r_i-\bm r_j\rVert_2\right).
\end{equation}
The corresponding isotropic Heisenberg Hamiltonian is
\begin{equation}
H(\bm r)
=
\sum_{i<j}J_{ij}
\left(
X_iX_j+Y_iY_j+Z_iZ_j
\right).
\end{equation}
This Hamiltonian is invariant under global spin rotations and therefore matches the inductive bias of the proposed architecture.

Labels are generated by exact diagonalization. For each spatial configuration, we construct the Hamiltonian matrix and take its lowest eigenvalue as the ground-state energy:
\begin{equation}
E_0(\bm r)=\lambda_{\min}(H(\bm r)).
\end{equation}
To match the finite range of ancilla measurement outputs, the energy labels are linearly normalized to $[-1,1]$ and denoted by $\widetilde E_0(\bm r)$. The regression loss is
\begin{equation}
\mathcal L_{\mathrm{reg}}(\bm\theta)
=
\frac{1}{B}\sum_{b=1}^{B}
\left(
f_{\bm\theta}(x_b)-\widetilde E_0(x_b)
\right)^2.
\end{equation}
The dataset contains $600$ samples, with $480$ for training and $120$ for testing. The model uses $N=4$ data qubits, $m=1$ shared ancilla qubit, and circuit depth $4$. We use Adam with learning rate $2\times10^{-2}$.

Figure~\ref{fig:heisenberg-en} shows the training curves averaged over $5$ random seeds. The training MSE stabilizes near $0.04$, and the test MSE reaches approximately $4\times10^{-2}$. The close curves show no marked train--test separation at this data scale. The result extends the finite-size learning evidence from classification to a continuous rotationally invariant spectral target.

\begin{figure}[t]
  \centering
  \includegraphics[width=0.86\textwidth]{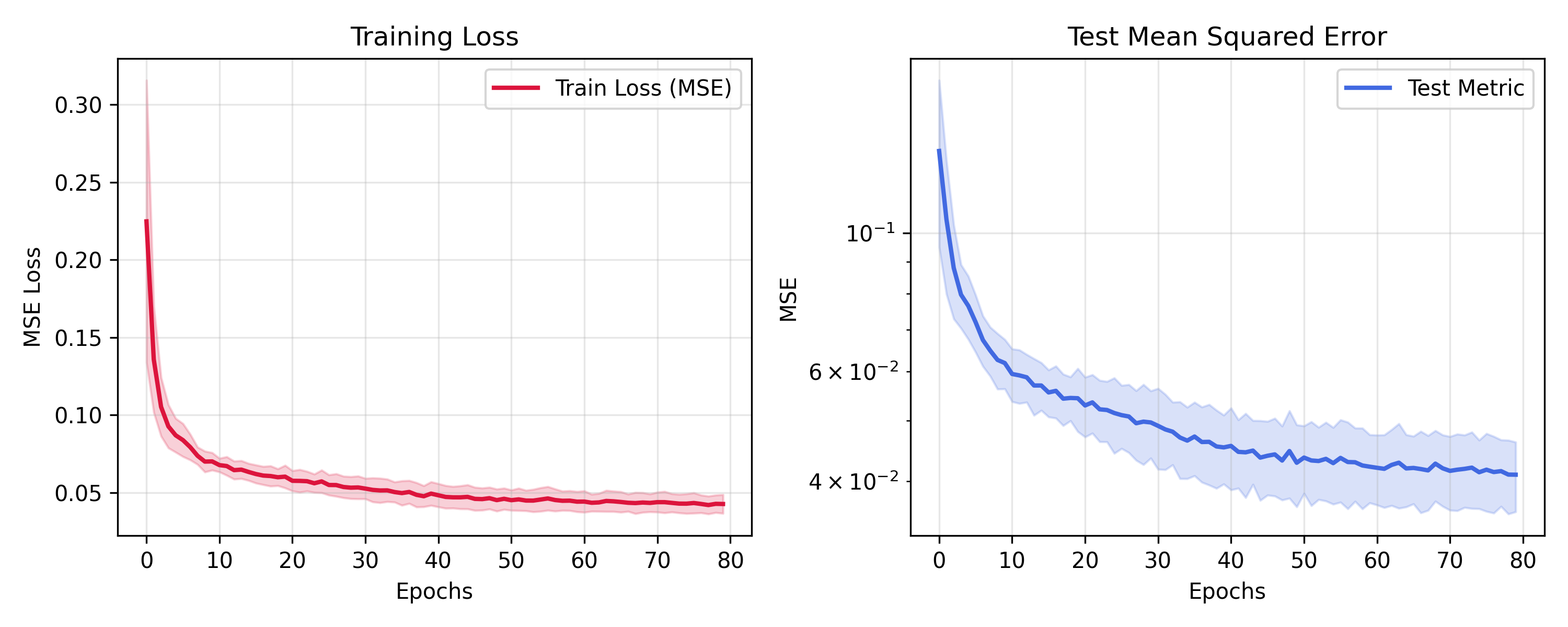}
  \caption{Training loss and test MSE for Heisenberg ground-state energy prediction.}
  \label{fig:heisenberg-en}
\end{figure}

\subsection{Interpretation of the Finite-Size Results}

The simulations connect the algebraic restriction to observable behavior under exact finite-size state-vector access. The gradient experiments compare initialization statistics across circuit families, while the point-cloud and Heisenberg tasks test two types of invariant targets. The latter has a direct physical correspondence: geometry changes the couplings of an $SU(2)$-invariant Hamiltonian, and the label is its lowest eigenvalue. Together, the experiments form an $SU(2)$ case study of the general architecture. Other groups require their own commuting sector decompositions and label-extraction circuits.

Natural extensions include enumerating the realized Lie algebra for small $N$, comparing it with Eq.~\eqref{eq:dim-bound-en}, and studying additional spin sizes, coupling graphs, and observables. Hardware-oriented evaluations can incorporate synthesized label extraction, finite-shot estimation, and noise channels. The relevant structural quantities are the retained-mode count $K_\ell$, ancilla size $m$, nonzero joint projectors $P_{\bm a}$, and the gradient or quantum Fisher information rank.

\section{Conclusion}

We introduced a hierarchical ancilla-controlled architecture in which commuting invariant-sector projectors select parameterized operations on a shared ancilla register. The data register retains the symmetry representation, while the non-Abelian trainable dynamics remains ancilla local. Coherent label extraction and uncomputation provide an operator-level realization; separate bounds describe the variational parameter count and the symbolic cost of label extraction, predicates, and controlled ancilla evolution.

The global circuit decomposition shows that each compatible joint sector follows an ordered ancilla-response path and that the scalar output is a weighted sum of these path responses. The number of realized paths, the ancilla dimension, and the available ancilla generators jointly regulate capacity, with parameter sharing connecting paths that carry the same layer labels. The group-independent containment theorem places the dynamical Lie algebra inside a direct sum of ancilla Lie-algebra blocks indexed by the same sectors. Constant ancilla size and a constant retained-mode count yield polynomial dimension growth along a logarithmic-depth hierarchy.

Nested particle-number and parity projectors demonstrate that the algebraic construction extends beyond rotations. For $SU(2)$, a fixed Clebsch--Gordan tree recursively supplies compatible atomic sectors and produces a rotation-equivariant circuit with rotation-invariant scalar readout.

Across the studied finite sizes, this realization shows slower initialization-gradient decay than the generic and conventional rotationally equivariant baselines. It also fits sparse rotation-invariant classification targets and geometry-dependent Heisenberg ground-state energies. These results support trainable Lie-algebra restriction as a structural strategy for initialization trainability. Adaptive sector selection, realized-dimension enumeration, and elementary-gate synthesis of coherent label extraction are direct directions for further study.

\begin{acknowledgements}
This work was supported by the National Natural Science Foundation of China under Grant No.~62271265.
\end{acknowledgements}

\section*{Data Availability}

The ModelNet40 dataset used in this work is publicly available from its original source. The processed data, simulation results, and code used to generate the numerical results are available from the corresponding author upon reasonable request.

\bibliographystyle{quantum}
\bibliography{reference/Ref_Article,reference/Ref_Conference,reference/Ref_Book}

\end{document}